\documentclass[twocolumn,showpacs,amsmath,amssymb]{revtex4}
\usepackage{graphicx}
\usepackage{dcolumn}
\usepackage{subfigure}
\usepackage{bm}
\usepackage{amsmath}
\usepackage{amsthm}
\usepackage{multirow}

\begin{document}

\title{A generic construction of quantum-oblivious-key-transfer-based private query with ideal database security and zero failure }
\author{Chun-Yan Wei$^{1,2}$}
 \author{Xiao-Qiu Cai$^{2}$}
 \author{Bin Liu$^{3}$}
 \author{Tian-Yin Wang$^{1}$}
\author{Fei Gao$^{2}$}
 \email{gaofei_bupt@hotmail.com}

\affiliation{%
 $^{1}$College of Mathematical Science, Luoyang Normal University, Luoyang, 471022, China\\
 $^{2}$State Key Laboratory of Networking and Switching Technology, Beijing University of Posts and Telecommunications, Beijing, 100876, China\\
 $^{3}$College of Computer Science, Chongqing University, Chongqing 400044, China}

\begin{abstract}
Higher security and lower failure probability have always been people's pursuits in quantum-oblivious-key-transfer-based private query (QOKT-PQ) protocols since Jacobi \emph{et al}. [Phys. Rev. A 83, 022301 (2011)] proposed the first protocol of this kind. However, higher database security generally has to be obtained at the cost of a higher failure probability, and vice versa. Recently, based on a round-robin differential-phase-shift quantum key distribution protocol, Liu \emph{et al}. [Sci. China-Phys. Mech. Astron.58, 100301 (2015)] presented a private query protocol (RRDPS-PQ protocol) utilizing ideal single-photon signal which realizes both ideal database security and zero failure probability. However, ideal single-photon source is not available today, and for large database the required pulse train is too long to implement. Here, we reexamine the security of RRDPS-PQ protocol under imperfect source and present an improved protocol using a special ``low-shift and addition'' (LSA) technique, which not only can be used to query from large database but also retains the features of ``ideal database security'' and ``zero-failure'' even under weak coherent source. Finally, we generalize the LSA technique and establish a generic QOKT-PQ model in which both ``ideal database security'' and ``zero failure'' are achieved via acceptable communications.
\end{abstract}
\maketitle

\section{Introduction}

Private information retrieval \cite{PIR,PIR1}, allows a user Alice to retrieve an item (which is generally supposed to be a bit) $x_{i}$ from a database $x_{1}x_{2}\cdots x_{N}$ without disclosing the retrieval address $i$ to the database holder Bob (user privacy). Symmetrically private information retrieval (SPIR) \cite{SPIR} concerns one more requirement named ``database security'', that is, Alice should not get more information than her wanted item from database. As database security and user privacy are in conflict, the task of SPIR cannot be realized ideally even in quantum cryptography \cite{unideal}. Quantum private query (QPQ) \cite{GLM}, the quantum scheme for SPIR, generally relaxes the security as follows. Alice can elicit several items instead of the ideal requirement (i.e., just one item) from database, and Bob's attack to guess user's retrieval address will be detected with a nonzero probability (cheat-sensitivity). Earlier QPQ protocols \cite{GLM,VSL,O} based on oracle operations show great significance in theory, but they are not loss-tolerant and for large database the dimension of the oracle operation would be too high to implement.

In 2011, Jakobi \emph{et al}. \cite{J} proposed a private query protocol (J-protocol) based on SARG04 quantum key distribution (QKD) protocol \cite{SARG04}. Concretely, Alice and Bob first share a raw oblivious key $K^{r}$ via SARG04 protocol in the way that (1) Bob knows $K^{r}$ entirely, (2) Alice knows every bit with probability $p$ $(p<1)$, and (3) Bob does not know which bits are known to Alice. Then, they conduct a classical postprocessing to ensure that Alice only knows roughly one bit of the final key $K^{f}$. Finally, Bob uses $K^{f}$ to encrypt the database according to a shift claimed by Alice in order that she can extract her wanted item from the encrypted database. As the generation of quantum-raw-oblivious-key and the postprocessing actually compose a quantum oblivious key transfer (QOKT) protocol, this kind of private query can be called ``QOKT-based private query'' (QOKT-PQ).

QOKT-PQ has become a research hotpot today because it is loss-tolerant and can be used to large database. Gao \emph{et al}. \cite{Gao} generalized J-protocol and proposed a more flexible version in 2012. Panduranga Rao \emph{et al.} \cite{Rao} gave efficient modifications of the classical postprocessing. Then, Zhang \emph{et al}. \cite{Zhang} designed a counterfactual private query protocol. Wei \emph{et al}. \cite{Wei} proposed a protocol in which the user can obtain a multi-bit block from database in one query. Then, Chan \emph{et al}. \cite{Chan} presented a fault-tolerant protocol. Gao \emph{et al}. \cite{Gao1} exhibited an effective attack on several postprocessings and gave an error-correction method for the oblivious key. Recently, Wei \emph{et al}. \cite{Wei1} proposed a practical protocol which can resist the joint-measurement attack.

Though many obstacles for the application of QOKT-PQ have been eliminated in succession, a troublesome stalemate still exists in most protocols, i.e., the ``trade-off'' relationship between database security and failure probability. Concretely, better database security has to be obtained at the cost of a higher failure probability, and vice versa. For example, when a $10^5$-bit database is concerned in J-protocol, Alice can obtain 6.10 bits averagely from database in one query with the failure probability being 0.002 if $k=7$ (see TABLE 1 in Ref. \cite{J} ), but if $k=8$ is selected to reduce the number of Alice's known bits to $10^5\times 0.25^8=1.53$, the failure probability will increase to $(1-0.25^8)^{10^5}=21.74\%$. Luckily, based on a novel round-robin differential-phase-shift (RRDPS) QKD protocol \cite{RRDPS}, Liu \emph{et al.} \cite{Liu} designed a private query protocol (RRDPS-PQ protocol) which breaks this stalemate. First, the number of database items an honest user can obtain is always one, which offers ideal database security (in previous protocols even an honest Alice generally can obtain several items instead of the ideal one item in order that the failure probability can be restricted to a small value). Second, the failure probability is always zero, meaning that the protocol will succeed all the time if the noise is ignored.

Despite these achievements, some problems still need to be further studied nowadays.

First, though the RRDPS-PQ protocol is perfect in theory, there are some obstacles in application due to the gap between ideal apparatuses and practical ones. (1) Ideal single-photon source is required in this protocol, but it is not available today. Weak coherent source (WCS) is generally used to replace it, then some pulses would inevitably contain multiple photons, which has been proven to open a serious security loophole for QKD \cite{PNS1,PNS2}. Therefore, the security of RRDPS-PQ protocol need to be further analyzed if WCS is used, and if it is badly affected, its essential to find an effective method to solve this problem. (2) A train of $N+1$ coherent pulses is needed when the database length equals to $N$, then the length of the pulse train would become huge for large database. With the increase of $N$, the implementation of the protocol would become too complex to realize with high speed and good stability (similar to the experimental setup in Ref. \cite{shiyan}, too many optical delays and beam splitters need to be placed accurately). In this sense, RRDPS-PQ protocol is difficult to directly implement to large database. How to remedy this flaw?

Second, the universal method to break the stalemate between ``database security'' and ``failure probability'' in general QOKT-PQ protocols is still awaited though it was solved in RRDPS-PQ protocol when an ideal single-photon source is used. However, ideal single-photon source is not available today. More importantly, most QOKT-PQ protocols are based on the BB84-like protocols (e.g., the SARG04 protocol used in \cite{J}) rather than the RRDPS-QKD protocol, hence finding a universal method to resolve this stalemate in general QOKT-PQ protocols is an urgent task today.

Here, we focus on the above problems, and the following contributions are achieved.
\begin{itemize}
\item We analyze the security of RRDPS-PQ protocol under WCS, and find that Alice can obtain more items than expected via multiple queries due to the existence of multiple photons, which clearly violates the ideal database security.
\item We give an improved RRDPS-PQ protocol with a special ``low-shift and addition'' (LSA) technique, which remedies the two flaws of the original protocol. Concretely, it can be used to query from large database because it uses shorter pulse trains (which are easier to prepare and implement) as information carriers. More importantly, it retains the good features of the original RRDPS-PQ protocol even under WCS, that is, it has ideal database security and zero failure probability despite the existence of multiple photons. Note that LSA is very crucial here because it for the first time fulfills a very difficult task in postprocessing, i.e., compressing Alice's information on the final key to one bit quickly and meanwhile eliminating the failure probability. For simplicity, we call the feature of ``ideal database security and zero failure probability" as ``IDS-ZF'' from now on.
\item Inspired by the perfect performance of LSA, we use it to construct a generic QOKT-PQ model with the ``IDS-ZF'' feature. That is, a universal method to realize both ideal database security and zero-failure probability in general QOKT-PQ protocols is exhibited.
\end{itemize}
The rest of this paper is organized as follows. We analyze the security of RRDPS-PQ protocol under WCS in Section 2, give an improved protocol in Section 3 and then analyze its security in Section 4. In Section 5, we exhibit a generic construction of QOKT-PQ protocols with the ``IDS-ZF'' feature. Finally, a brief conclusion is given in Section 6.

\section{Analysis of RRDPS- PQ protocol with WCS}
\subsection{Review of RRDPS-PQ protocol}
When an $N$-bit database $x_{1}x_{2}\cdots x_{N}$ is concerned, the RRDPS-PQ protocol \cite{Liu} is as follows.
\begin{itemize}
\item[1)]Bob sends Alice a single-photon state $|\Psi_{S}\rangle$ of $N+1$ pulses according to a randomly chosen $(N+1)$-bit string $S=s_{0},s_{1},\cdots,s_{N}$ with $s_{k}\in\{0,1\}$ for $k=0,1,2,\cdots,N$. That is,
$|\Psi_{S}\rangle=\frac{1}{\sqrt{N+1}}\sum_{k=0}^{N}(-1)^{s_{k}}|k\rangle$,
where $|k\rangle$ means that the photon is in the $k$-th pulse.

\item[2)] When receiving $|\Psi_{S}\rangle$, Alice chooses a random $r\in\{1,2,\cdots,N\}$, splits each pulse by a half beam splitter and shifts the pulses in one light path by $r$. Then she uses the interference circuits to randomly get one of the values
$\{s_{j}\oplus s_{j'}\}_{j=0}^{N}$,
where $ j'=j+r \mod (N+1)$, and $\oplus$ denotes summation modulo 2.
Suppose the value Alice finally gets is $s_{t}\oplus s_{t'}$.
\item[3)] Alice publishes $t$ and then they share an $N$-bit oblivious key $K$
\begin{center}
 $s_{t}\oplus s_{0}$, $s_{t}\oplus s_{1}$, ..., $s_{t}\oplus s_{t-1}$, $s_{t}\oplus s_{t+1}$,...,$s_{t}\oplus s_{N}$.
\end{center}
Note that Alice only knows the bit $s_{t}\oplus s_{t'}$ in it.
\item[4)] Suppose Alice knows the $j$-th bit in $K$ and wants the $i$-th item in database, then she claims a shift $s=i-j$.
\item[5)] Bob encrypts his database with $K$ shifted by $s$, and then sends the encrypted database to Alice.
\item[6)] Alice recovers the wanted item by her known bit in $K$.
\end{itemize}

Owing to the ideal information carrier, i.e., a single photon split into $L+1$ pulses, Alice always detects one photon in her site and obtains exactly one bit in $K$ if the noise and dark counting are ignored. It provides two fantastic merits. First, the classical postprocessing to reduce Alice's knowledge on the oblivious key can be eliminated, which obviously reduces the communications. Second, it realizes the ``IDS-ZF'' feature, that is, an honest user can only get her wanted item from database, and the failure probability is always zero. However, these merits can hardly do without the ideal carrier. As ideal single photon source is not available with present technology, we now consider using WCS to replace it. In this case, there would inevitably have multiple photons in the pulse train. As we know, the RRDPS-QKD protocol is highly robust against similar attacks \cite{PNSR} induced by the imperfect source, then what if the WCS is used in the RRDPS-PQ protocol?

\subsection{Security of RRDPS-PQ protocol under WCS}
The existence of multiple photons gives Alice a chance to obtain more information about the key $K$ when weak coherent source is used in RRDPS-PQ protocol. For example, if the pulse train contains two photons and Alice splits it into two paths to measure them according to a random shift, e.g., 2, the two photons will be in different pulses with high probability. Then, Alice may obtain two phase differences, e.g., $s_{2}\oplus s_{4}$, $s_{3}\oplus s_{5}$. As there is no classical postprocessing in this protocol, if Alice announces the time slot she has detected the photon is 2, the final key would be
\begin{center}
$r_{1}=s_{2}\oplus s_{0}$, $r_{2}=s_{2}\oplus s_{1}$, $r_{3}=s_{2}\oplus s_{3}$, $r_{4}=s_{2}\oplus s_{4}$, $r_{5}=s_{2}\oplus s_{5}$, $r_{6}=s_{2}\oplus s_{6}$, $\cdots$,  $r_{N}=s_{2}\oplus s_{N}$.
\end{center}
Clearly, Alice knows the bit $r_{4}$ and the parity of $r_{3}$ and $r_{5}$ because $r_{3}\oplus r_{5}=s_{3}\oplus s_{5}$. As shown in Ref. \cite{Gao1}, knowing parity information of the final key bits can help Alice to obtain more database items in multiple queries. For example, if Alice queries the $4$-th bit $x_{4}$ from database, she would obtain $x_{4}$ and the parity information $x_{3}\oplus x_{5}$ after receiving the encrypted database $c_{1}c_{2}\cdots c_{N}$ because $x_{4}$=$c_{4}\oplus r_{4}$ and $x_{3}\oplus x_{5}$=$c_{3}\oplus r_{3}\oplus c_{5}\oplus r_{5}$. Then, if she retrieves $x_{3}$ in the next query, she will at least obtain $x_{3}$ and $x_{5}$. That is, Alice obtains at least three items in two queries.

We now estimate the influence of the above attack by doing simulations similar to that in Ref.\cite{Gao1} to check how many queries are needed for Alice to obtain one whole database. Consider the most advantageous case for Alice, that is, she claims a successful detection only when obtaining no less than two phase differences. This attack is possible when channel loss is high. Therefore in each query, Alice knows one bit as well as the parity of two other bits in the key, and she is allowed to select an optimal shift so that the number of her known database items is the highest. We execute this simulation 25 runs and find that Alice can obtain one $10^{4}$-bit database totally by on average 5838 queries. It is clearly violates the ideal database security, then how to retain this requirement under imperfect source?

\section{An improved RRDPS-PQ protocol}

As analyzed above, our main concerns to improve the RRDPS-PQ protocol are as follows. (1) If the source is imperfect, how to retain the merits of ``ideal database security'' and ``zero failure'' (IDS-ZF)? (2) How to query from large database? We now give an improved RRDPS-PQ protocol as follows.
\begin{figure*}
  \centering

  \includegraphics[scale=0.89, bb=20 170 565 550]{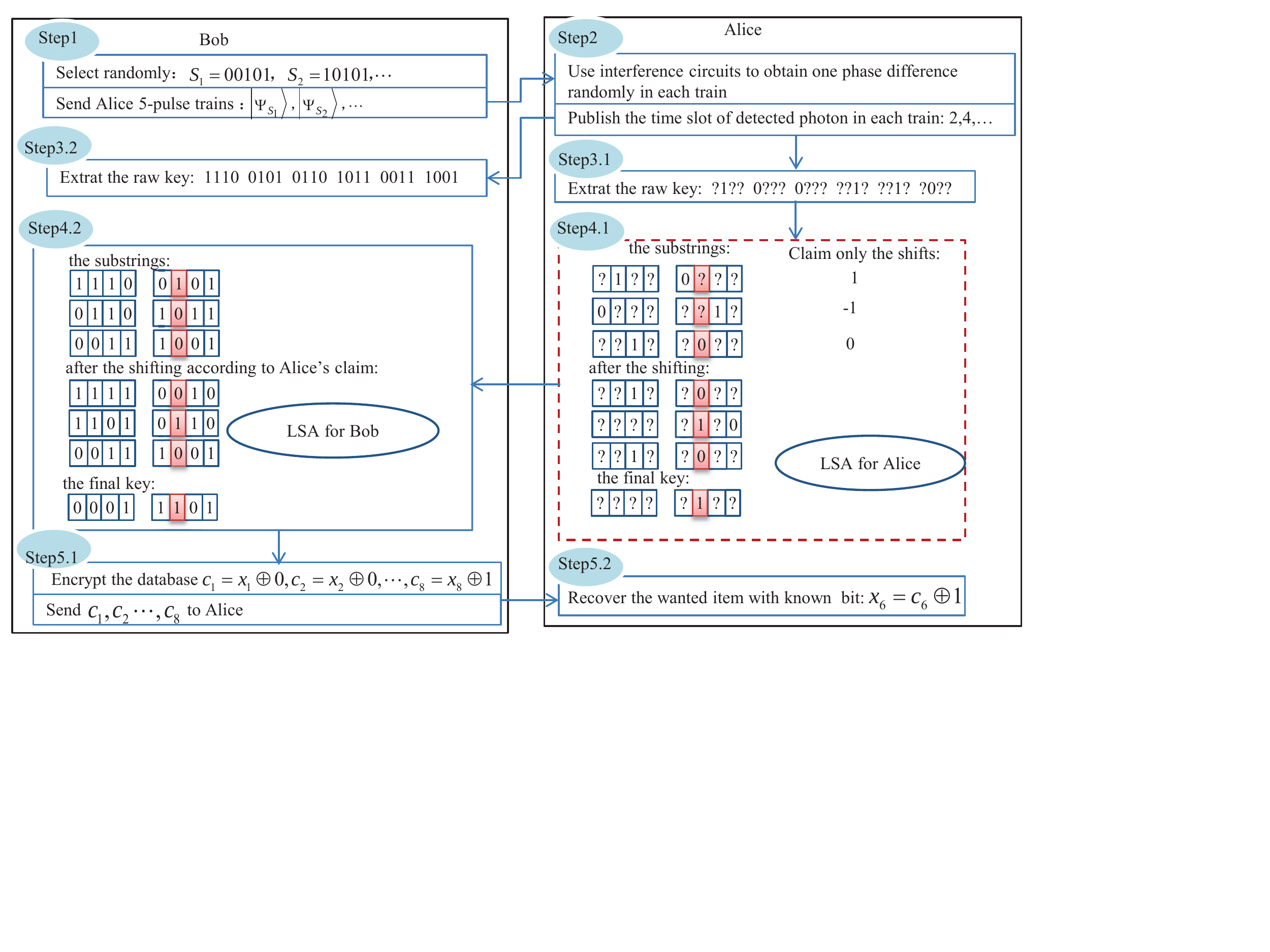}
  \caption{A sketch of the improved protocol. Here, $N=8$ and $l=4$. Suppose Alice wants the $6$-th item $x_{6}$ in database. Bob sends many $5$-pulse trains to Alice, then Alice uses the interference circuits to obtain one of the phase differences. As a result, Alice knows exactly one bit in each $4$ raw key bits (symbols ``?'' denote Alice's unknown bits). Therefore, in the ``low-shift and addition'' (LSA) phase (Step4.1), the shift Alice selects for each substring can be nicely located between -3 and 3. Here, she claims shifts 1,-1 and 0 for the three substrings so that her known bits can be located in the $6$-th positions (the red rectangles) after the shifting. Clearly, she can obtain the $6$-th final key bit and then use it to extract $x_{6}$ directly. }
\end{figure*}

Suppose $l$ is a small positive integer so that the train of $l+1$ coherent pulses can be easily prepared and implemented with current technology. For simplicity, we assume that the database length $N$ is divisible by $l$, otherwise we can achieve it by adding for example some bits ``0''s at the end of the database. A sketch of this protocol is shown in Fig.1.

\begin{itemize}

 \item[1)] Bob sends many $(l+1$)-pulse trains to Alice, and each train is a single-photon of $l+1$ sequential coherent pulses with each pulse modulated with phase $0$ or $\pi$ randomly. Concretely, when preparing one pulse train, Bob randomly selects one $(l+1)$-bit string $S=s_{0}s_{1}\cdots s_{l}$ with $s_{k}\in\{0,1\}$ for $k=0,1,2,\cdots,l$, and prepares the state $|\Psi_{S}\rangle=\frac{1}{\sqrt{l+1}}\sum_{k=0}^{l}(-1)^{s_{k}}|k\rangle$,where $|k\rangle$ means that the photon is in the $k$-th pulse.

 \item[2)] For each received $(l+1)$-pulse train, Alice splits it by a half beam splitter and shifts the pulses in one light path according to a random value chosen from $\{1,2,\cdots,l\}$. Suppose one pulse train is prepared with string $S=s_{0}s_{1}\cdots s_{l}$ and the shift Alice selects for it is $r$, then Alice uses the interference circuits to randomly get one of the values
$\{s_{j}\oplus s_{j''}\}_{j=0}^{l}$,
where $\ j''=j+r \mod (l+1)$, and $\oplus$ denotes summation modulo 2. This phase is similar to that in the original protocol \cite{Liu}.
Alice publishes $t$ if the value she finally gets is $s_{t}\oplus s_{t''}$. As a result, Bob knows all of the following $l$ bits
\begin{center}
$s_{t}\oplus s_{0}$, $s_{t}\oplus s_{1}$, ..., $s_{t}\oplus s_{t-1}$, $s_{t}\oplus s_{t+1}$,...,$s_{t}\oplus s_{l}$,
\end{center}
while Alice knows only one bit $s_{t}\oplus s_{t''}$. If Alice claims that she has not detected any photon, they discard this pulse train.
\item[3)] Alice and Bob should transmit enough pulse trains so that they can share a $kN$-bit ($k$ is a security parameter and we will discuss it later) raw key in total.
\item[4)]They cut the raw key into $k$ $N$-bit substrings. Suppose Alice wants the $i$-th item in database, then she claims a shift $S_{j}\in \{-(l-1),-(l-2),\cdots,l-1\}$ for the $j$-th $(j=1,2,\cdots,k)$ substring so that she knows the $i$-th bit after the shifting. These shifted substrings are added bitwise to obtain a final key. Obviously, Alice knows the $i$-th final key bit. Unlike previous QOKT-PQ protocols, proper shift for each substring here can be found in a small range $\{-(l-1),-(l-2),\cdots,l-1\}$ because Alice always knows one bit in each $l$ raw key bits, hence this phase can be called ``low-shift and addition '' (LSA).
\item[5)] Bob sends the database encrypted with the final key to Alice. Clearly, Alice can recover the $i$-th item with the $i$-th final bit she knows.
\end{itemize}

\section{Analysis of the improved protocol}

Similar to the analysis in the original RRDPS-PQ protocol, the user privacy in the improved protocol is also assured in the sense of cheat-sensitivity. That is, if by some attacks such as sending fake states and so on, Bob gets some effective information on which item Alice wants, he might give Alice the wrong value of the queried one, which can be discovered later. Here, we only discuss the database security as well as some other features of the improved protocol.
\subsection {The ``IDS-ZF'' feature of the improved protocol}
First, we show that the improved protocol can achieve the ``IDS-ZF'' feature even under WCS, retaining the merits of the original protocol. As Alice knows exactly one bit in each $l$ raw key bits if ideal single-photon source is used, and also knows some parity information if WCS is used (similar to that in Section 2.2), we only need to prove this character under WCS.

\newtheorem{theorem}{Theorem}
\begin{theorem}[IDS-ZF]

The improved protocol can achieve the ``IDS-ZF'' requirement even under weak coherent source (WCS).
\end{theorem}

\begin{proof}
Suppose $\mu$ is the mean photon number of the WCS, then each pulse train contains $m$ photons with probability $p_{m}=\frac{\mu^{m}}{m!}e^{-\mu}$. Consider the most advantageous case for Alice, that is, she can obtain $m$ phase differences if there are $m$ photons in this train. Therefore in each $l$ raw key bits Alice knows one bit and $m-1$ $(m\geq1)$ parity values of some other bits with probability $p_{m}'=p_{m}/\sum_{i=1}^{\infty}p_{i}$. Assume Alice wants the $i$-th item in database

 In the LSA, we assume that after the adding of the first $k$ (which is a positive integer) substrings Alice knows $n_{k}$ bits and $m_{k}$ parity-correlated bits in the output string $S$.  Clearly,
 $n_{k+1}\leq n_{k}$ and $m_{k+1}\leq m_{k}$. Alice can always know at least one bit in $S$ because the shifts are selected by herself, that is, $n_{k}\geq1$ holds for any positive integer $k$ and the failure probability will always be 0.

 We now show that $n_{k}=1$ and $m_{k}=0$ can be achieved quickly in LSA, which assures the ideal database security of this protocol. First, Alice has to choose a suitable low shift for the $(k+1)$-th substring so that she knows the $i$-th bit after shifting, and the other $n_{k}-$1 known bits would be kept known after the next adding with probability $\frac{1}{l}$, that is
 \begin{equation}
 n_{k+1}=1+(n_{k}-1)/l.
 \end{equation}
Second, the known parity-correlated bits must be the sum of known bits or known correlated bits. As the probability of ``unknown and uncorrelated'' bits in the raw key is
\begin{center}
$p''=\frac{l-1}{l}p_{1}'+\frac{l-3}{l}p_{2}'+\cdots\geq\frac{l-1}{l}p_{1}'$,
\end{center}
we have
 \begin{equation}
 m_{k+1}\leq m_{k}\cdot (1-p'')\leq m_{k}\cdot (1-\frac{l-1}{l}p_{1}').
 \end{equation}
 Note that $p_{1}'$ is larger than 0.5 because $\mu$ is generally significantly smaller than 1, hence Eqs. (1-2) show that $n_{k}=1$ and $m_{k}=0$ can be achieved quickly with the increase of $k$.

 Finally, we prove that ideal database security can be achieved even when Alice selects to obtain as many items as possible rather than one predetermined item. That is, for each substring Alice selects an optimal shift from $\{-(l-1),-(l-2),\cdots,(l-1)\}$ so that after the adding the number of her known bits is the highest. We check the case that $n_{k+1}=n_{k}$ first. Suppose $S_{i}$ is the set of the substrings which can retain $n_{k+1}=n_{k}$ for the shift $i$ in the next adding. Obviously, $S_{i}$ is determined by the intervals of the $n_{k}$ bits Alice needs to know. For example, when $l=4$ and the output of the first $k$ substrings is $S=$``???? ?0?? ??1? ???? ???0'' (``?'' is Alice's unknown bit),
 \begin{center}
 $S_{1}=\{$`\#\#\#\# d??? ?d?? \#\#\#\# ??d?'$\}$.
 \end{center}
 Here, `d' is Alice's known bit, and `\#' can be Alice's known or unknown bit. That is, $n_{k+1}=n_{k}$ can be achieved if the $(k+1)$-th substring is anyone of $S_{1}$ and its shift is 1. Clearly for any $i\in\{-(l-1),-(l-2),\cdots,(l-1)\}$, $S_{i}$ appears with probability no more than $\frac{1}{l^{n_{k}}}$, that is, $p(S_{i})\leq \frac{1}{l^{n_{k}}}$. Therefore the probability that $n_{k+1}=n_{k}$ can be written as

 \begin{eqnarray}
         p(n_{k+1}=n_{k})=p(S_{-(l-1)}\bigcup S_{-(l-2)}\bigcup\cdots S_{(l-1)})\nonumber\\
         \leq \sum\limits_{i=-(l-1)}^{l-1} p(S_{i})\leq \frac{2l-1}{l^{n_{k}}}{\quad\quad\quad\quad\!\;\:\,}\nonumber
 \end{eqnarray}

 When $n_{k}\geq2$, we have
\begin{equation}
p(n_{k+1}<n_{k})=1-p(n_{k+1}=n_{k})\geq(\frac{l-1}{l})^{2}.
\end{equation}
That is, once one more substring is shifted and added in LSA, Alice's known bits will strictly decrease with probability no less than $(\frac{l-1}{l})^{2}$ as long as $n_{k}\geq2$. Therefore, $n_{k}=1$ can be achieved if proper $k$ is selected.

\end{proof}

Now, to show the high efficiency of the LSA in compressing Alice's knowledge on the oblivious key and estimate the WCS's influence on database security, we do simulations to choose proper value of $k$ (i.e., the number of substrings combined in LSA to obtain a final key) for both honest and malicious Alice to ensure ideal database security. Here, by ``honest'' we mean that Alice does not try to illegally obtain more database items though she actually can obtain some parity information when multiple photons exist, that is, she ignores the parity information and executes the protocol faithfully. For honest Alice, the simulation is as follows.
\begin{itemize}

\item[(a)] Define database length $N$, the length of pulse train $l$, retrieval address $i$ and set $k=1$.

\item[(b)] (I) One $l$-bit block in which Alice knows exactly one bit is generated. (II) Repeat (I) until an $N$-bit string $S$ is generated and record the known/unknown bits in it.

\item[(c)] $S$ is shifted by a value $s\in\{-(l-1),-(l-2),\cdots,l-1\}$ according to the retrieval address $i$.

\item[(d)] (I) Repeat step (b) to generate a new $N$-bit string $S'$; (II) $S'$ is shifted by a value $s\in\{-(l-1),-(l-2),\cdots,l-1\}$ according to the retrieval address $i$, then $S$ and $S'$ are bitwise added. Denote the output of the adding as $S$ and record the known/unknown bits in it, then set $k=k+1$.

\item[(e)] Repeat step (d) until Alice knows only one bit in $S$, then output the value of $k$.
\end{itemize}

For malicious Alice, the simulation is as follows.
\begin{itemize}
\item[(a)] Define $N$, $l$, $\mu$ and set $k=1$. Here, $\mu$ is the mean photon number of the WCS. Consider the most advantageous case for Alice, that is, she can obtain $m$ phase differences if there are $m$ photons in this train and she announces a successful detection only when she obtains no less than two phase differences. Therefore in each $l$ raw key bits Alice knows one bit and $m-1$ $(m\geq2)$ parity values of some other bits with probability $p_{m}'=p_{m}/\sum_{i=2}^{\infty}p_{i}$.

\item[(b)] (I) One $l$-bit block with respect to $\mu$ is generated. Concretely in this block, Alice knows exactly one bit and $m-1$ parity values of some other bits with probability $p_{m}'$ for $m=2,3,\cdots$. (II) Repeat (I) until an $N$-bit string $S$ is generated, and record the known, unknown and parity-correlated bits in it.

\item[(c)] (I) Repeat (b) to generate a new $N$-bit string $S'$; (II) Optimal shifts for $S$ and $S'$ are selected in $\{-(l-1),-(l-2),\cdots,l-1\}$ so that after the adding of the shifted $S$ and $S'$ the number of Alice's known bits is the highest. (III) Denote the output of the adding as $S$ and record the known, unknown and parity-correlated bits in it, then set $k=k+1$.

\item[(d)] (I) Repeat (b) to generate a new $N$-bit string $S'$; (II) An optimal shift $s\in\{-(l-1),-(l-2),\cdots,l-1\}$ for $S'$ is selected so that after the adding of $S$ and the shifted $S'$ the number of Alice's known bits is the highest. (III) Denote the output of the adding as $S$ and record the known, unknown and parity-correlated bits in it, then set $k=k+1$.

\item[(e)] Repeat step (d) until Alice knows only one bit in $S$, then output the value of $k$.
\end{itemize}
One concrete simulation for malicious Alice is given in Fig.2, which shows that both the bits and the parity information Alice knows are reduced quickly with the increase of $k$. When $k=6$, Alice knows only one final bit and the ideal database security is achieved. To give a sufficiently secure $k$, we choose proper $k_{H}$ for honest Alice and $k_{M}$ for malicious one to ensure ideal database security by executing the above simulations 25 runs for $l=8$, and collect the maximum, minimum, mean and standard deviation ($std$) of $k_{H(M)}$ (see Table 1) for comparison. Ideal database security can be achieved quickly for both honest and malicious Alice though the values of $k_{M}$ are generally a little larger than that of $k_{H}$. Note that the abilities of malicious Alice are sufficiently enlarged because her attack would be detected by Bob owing to the superhigh rate of lost photons. Generally $k$ can be set with the mean value of $k_{M}$ to realize the ideal database security, and for insurance purposes, $k$ can be set with the maximum value of $k_{M}$. As shown in Table 1, when $N=10^{4}$, $l=8$ and $\mu=0.1$, $k=8$ is large enough to ensure ideal database security for both honest and malicious Alice.

\begin{figure}
  \includegraphics[scale=0.62,trim=30 32 0 0,clip]{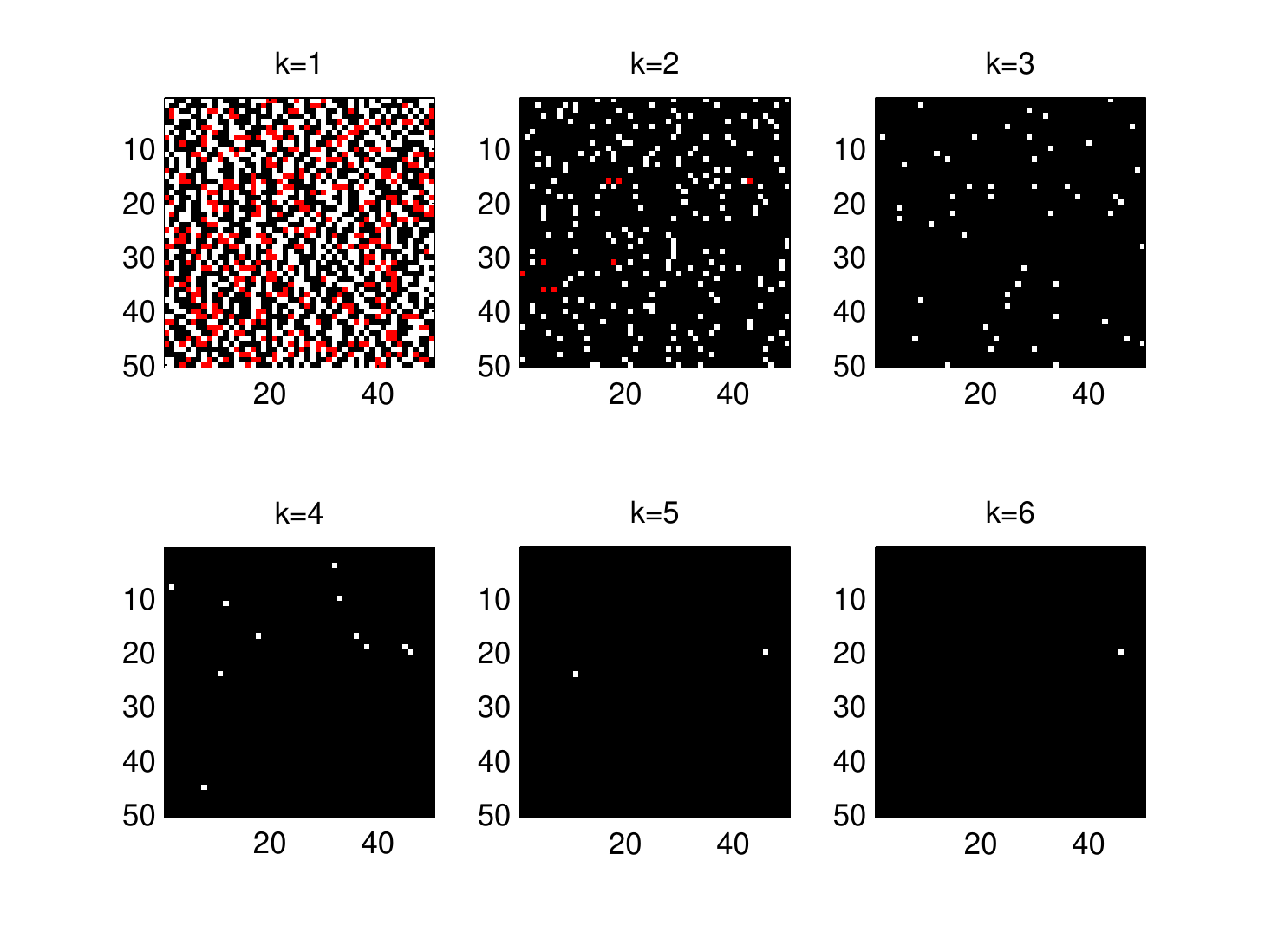}\\
  \caption{A simulation of the LSA to reduce malicious Alice's knowledge on the final key to only one bit under WCS.  Here, $l=8$, $N=2500$ and $\mu=0.1$. We draw a picture of the database where any item is represented by a square, and dye the unknown items as black, the known ones white, and the parity-correlated ones red. Clearly, Alice's known bits as well as the parity information are reduced quickly with the increase of $k$.  When $k=6$, Alice knows only one bit in the final key, which indicates the ideal database security of this protocol. }
\end{figure}

\begin{table}
 \centering
  \caption{For different $N$, proper choices of ${k_{H(M)}}$ for honest (malicious) Alice. in 25 simulations with $l=8$ and $\mu=0.1$.
}
\begin{tabular}{c|c|cccc}
\hline\hline
 \multicolumn{2}{c|}{$N$}& 900& 2500& $10^4$
 &90000\\
\hline
 \multirow{4}*{$k_{H}$}
  & maximum & 5 & 7 & 7&8\\
  & minimum & 3 & 4 & 4 &5\\
  & mean & 3.96& 4.96 & 5.28 &6.32 \\
  & std & 0.5385 & 0.8406 & 0.6782&0.6904\\
  \hline
 \multirow{4}*{$k_{M}$}
  & maximum & 6& 7 & 8&9\\
  & minimum & 4 & 4 & 5&6\\
  & mean & 4.16 & 4.92 & 5.84 & 7.08\\
  & std & 0.4726 & 0.7024 & 0.8& 0.8426\\
\hline\hline
\end{tabular}

\end{table}
\subsection{Some other features of the improved protocol}
The improved protocol is more practical. It retains the ``IDS-ZF'' feature even under WCS, which means that it is robust against the attack induced by imperfect source. Moreover, it can be used to large database because the information carries, i.e., the $(l+1)$-pulse trains with $l$ being a small positive integer, are much easier to prepare and implement.

The improved protocol is loss-tolerant. Though Alice can announce a successful detection only when she detects multiple photons in one pulse train, her extra information will be compressed quickly in the LSA, and Bob generally can find this attack due to the superhigh rate of lost photons.

The improved protocol can resist the quantum memory attack. Alice cannot gain any benefit by storing the photons for a delayed measurement because Bob would not leak any information about the photons later.

\section{Generic model of QOKT-PQ protocols with ``IDS-ZF'' feature}

In previous years, appropriate postprocessing of the raw key has been regarded as the key phase to achieve both ``higher database security'' and ``lower failure probability''. On one hand, to realize higher database security, Alice's information on the oblivious key has to be further compressed by certain method such as bitwise adding. On the other hand, to reduce the failure probability, Alice should be allowed to select shifts for the substrings in order that she can know at least one final key bit. Actually, some sophisticated methods combining ``shifts selected by Alice'' and ``bitwise adding'' have been given in previous QOKT-PQ protocols.
\begin{itemize}

 \item Jacobi \emph{et al.} \cite{J} gave a sophisticated method to improve the database security, i.e., executing the QOKT phase (including ``quantum raw oblivious key distribution'' and ``bitwise adding'') twice to create two final keys in which Alice knows only several bits, and then combining them with a shift chosen by Alice. As a result, Alice knows exactly one bit finally with overwhelming probability, but the failure probability may increase dramatically. For example, when a $10^4$-bit database is concerned, the failure probability would be on average 0.087 (see Table 1 in \cite{J}) in one run of the QOKT, but as it need to be executed twice when the sophisticated method is used, at least one final key is totally unknown to Alice with probability $1-(1-0.087)^2=0.1664$, that is, the failure probability would increase to 16.64\%.

\item Gao \emph{et al.} \cite{Gao1} used the ``shift-addition" method to further compress Alice's information on some middle keys (which are actually oblivious keys) created by the error correction of the raw oblivious key. That is, Alice freely chooses a shift for each middle key before bitwise adding. It compresses Alice's known bits dramatically and makes the error-correction of the oblivious key more practical. However, as one concrete simulation of the shift-addition (see Table 2 in \cite{Gao1}) shows, it is very difficult to reduce Alice's knowledge on the final key to less than 3 bits even with a dramatic increase of $k$, that is, ideal database security is nearly impossible via acceptable communications.
\end{itemize}
Therefore, though in general QOKT-PQ protocols the stalemate between ``database security" and ``failure probability" has been eased by the ``shift-addition'' method, it is still far from ``well-resolved". Then, why it can be resolved perfectly in the improved RRDPS-PQ protocol? Actually in this protocol Alice knows at least one bit in each $l$ raw key bits, which can limit the shifts chosen by Alice to a small range $\{-(l-1),-(l-2),\cdots,(l-1)\}$. Note that ``low-shift'' is very important to realize ideal database security via ``acceptable communications'' because it greatly reduces Alice's advantage in selecting optimal shifts. As Eq.(3) shows, with probability no less than $(\frac{l-1}{l})^{2}$ Alice's known bits would be strictly decreased with the increase of $k$ as long as $n_{k}\geq2$. That is, it is easy to further compress Alice's known bits even when $n_{k}$ is small, which is quite different from previous `` shift-addition'' method (note that in \cite{Gao1}, even with a sharp increase of $k$, it is very difficult to reduce Alice's known bits when $n_{k}$=3 ).

A natural question is, can we use the LSA skill in general QOKT-PQ protocols to realize ideal database security as well as ``zero-failure'' quickly? As analyzed above, only when Alice knows at least one bit in each $l$ ($l\ll N$) raw key bits, the shift can be limited to a small range and the LSA can be used to realize the feature of ``IDE-ZF'' quickly.

\subsection{Generic QOKT-PQ model with the ``IDS-ZF'' feature}
  In this section, we generalize the LSA technique and give a generic QOKT-PQ model with the ``IDS-ZF'' feature as follows.

  For simplicity, we assume that the database length $N$ is divisible by $l$, otherwise we can achieve it by appending some bits at the end of the database.
  \begin{itemize}

  \item[1)] \emph{Quantum oblivious key distribution}. Alice and Bob share a raw oblivious key with approximately $kN/(1-(1-l)^{l})$ bits via certain QKD protocol ($k$ is a security parameter).

  \item[2)]\emph{Block-sifting}. To ensure that Alice knows at least one bit in each $l$ raw key bits, they divide the raw key into $l$-bit blocks, then Alice announces which blocks are totally unknown to her and should be discarded. The proportion of discarded blocks should be about $p_{discard}=(1-p)^{l}$. The remaining raw key should be composed of $kN$ bits, i.e., $k \frac{N}{l}$ blocks, and in each block Alice knows at least one bit.

  \item[3)]\emph{Low-shift and addition (LSA)}. Alice and Bob divide the raw key into $k$ substrings with equal length $N$. Then Alice announces a shift in $ \{-(l-1),-(l-2),\cdots,(l-1)\}$ for each substring according to her wanted item. These shifted substrings are bitwise added to obtain a final key.

  \item[4)]\emph{Retrieval}. Bob sends Alice the database encrypted with the final key in the manner of one time pad. Clearly, Alice can obtain the wanted database item by her known final bit.
\end{itemize}
\subsection{Security and parameters}

After the ``Block-sifting'', the probability of Alice's known bits in the raw key would be $p'=\frac{p}{1-(1-p)^{l}}$ and Alice knows at least one bit in each $l$ raw key bits. As a result, the generic QOPT-PQ model can be proved to realize the feature of ``IDE-ZF'' via similar method of Theorem 1. Hence, we omit the proof and only give the following theorem here.
\begin{theorem}
The generic ``QOKT-PQ'' model can realize both ideal database security and zero-failure probability (IDE-ZF).
\end{theorem}

Now we turn to examine how quickly the LSA can reduce Alice's information on the final key to one bit. To give sufficiently secure parameters, we also give Alice the advantage of choosing suitable retrieval address and optimal shifts here. Concretely, we do simulation to select proper $k$ as follows.
\begin{itemize}

\item[(a)] Define the parameters $N$, $p$, $l$, $n_{A}$. Here, $n_{A}$ is the number of database items Alice is expected to obtain. Set $k=1$ and $n'=\frac{N}{1-(1-p)^{l}}$.

\item[(b)] Key generation and Block-sifting. (I) One $n'$-bit string with respect to $p$ is generated. Concretely, Bob knows the string completely while Alice knows every bit with probability $p$. (II) Divide the string into $l$-bit blocks, discard the ones which are completely unknown to Alice. The remaining string $S$ should be composed of $N/l$ blocks. Record Alice's known/unknown bits in it.

\item[(c)] (I) Repeat step (b) to generate a new string $S'$; (II) Select optimal shifts from $\{-(l-1),-(l-2),\cdots,l-1\}$ for $S$ and $S'$ in order that after the adding of the shifted $S$ and $S'$ the number of known bits is the highest. (III) Denote the output of adding as $S$ and record Alice's known/unknown bits in it, then set $k=k+1$.

\item[(d)] (I) Repeat step (b) to generate a new string $S'$; (II) Select an optimal shift $s\in\{-(l-1),-(l-2),\cdots,l-1\}$ for $S'$ in order that after the bitwise adding of $S$ and the shifted $S'$ the number of known bits is the highest. (III) Denote the output of the bitwise adding as $S$ and record Alice's known/unknown bits in it, then set $k=k+1$.

\item[(e)] Repeat step (d) until Alice knows only $n_{A}$ bits in $S$, then output the value of $k$.
\end{itemize}

We set $p=0.25$ and simulate 100 runs for different $N$, $l$ and $n_{A}$, and then give proper choices of $k$ in Table 2. Here, $\overline{k}_{l}$ represents the average value of $k$ for block length $l$. We can find that both high database security (i.e., $n_A\leq3$ ) and ``ideal database security'' (i.e., $n_A$=1) can be achieved quickly. For example, when $N=10^4$ and $l=10$, the mean value of $k$ to ensure that Alice obtains only one final bit is 12.74, and the discarded blocks only occupy 5.63\%, which is clearly an acceptable communication.

\begin{table}
 \centering
  \caption{Proper $\overline{k}_{l}$ in the generic QOKT-PQ model for different block lengths $l$ and $p=0.25$ to ensue that Alice can obtain only $n_{A}$ database items.}\label{tab:g}
  \vspace{-1.8mm}
\begin{tabular}{ccccc|c}
\hline\hline\vspace{-1.8mm}\\
 \multicolumn{2}{c}{}& $n_{A}$=3& $n_{A}$=2& $n_{A}$=1& $p_{discard}$\\
\hline
  & $\overline{k}_{8}$ & 7.52 & 8.57 & 12.04& 0.1001\\
  $N=10^4$ & $\overline{k}_{10}$ & 7.32 & 8.48 & 12.74 & 0.0563\\
  & $\overline{k}_{16}$& 7.41& 8.78 & 16.13  &0.010\\
 \hline
  & $\overline{k}_{8}$  & 9.33& 10.29 & 13.43 &0.1001\\
  $N=10^5$& $\overline{k}_{10}$& 9.10 & 10.33 & 14.62 &0.0563 \\
  & $\overline{k}_{16}$ & 9.29 & 10.66 & 17.55 & 0.010\\
\hline\hline\vspace{-1.8mm}\\
\end{tabular}
\end{table}

To further emphasize the importance of ``low shift'' for ideal database security, we compare the effects of different shifts in compressing Alice's information in Fig.3. Concretely when $p=0.25$ , if the shifts are randomly chosen, it is very difficult to reduce the number of Alice's known bits to less than 5 even when $k$ has increased to 24 or much bigger values, which means that ideal database security is nearly impossible via acceptable communications. On the contrary, if the shifts are restricted to $\{-(l-1),-(l-2),\cdots,l-1\}$ with $l=8$, 10 or 16, the efficiency of compressing increases significantly and $k=16$ is large enough to achieve ideal database security. That is why the LSA technique can be used to perfectly address the troublesome stalemate of ``database security'' and ``failure probability'' in general QOKT-PQ protocol.

\begin{figure}
  \includegraphics[scale=0.62,trim=26 5 0 20 ,clip]{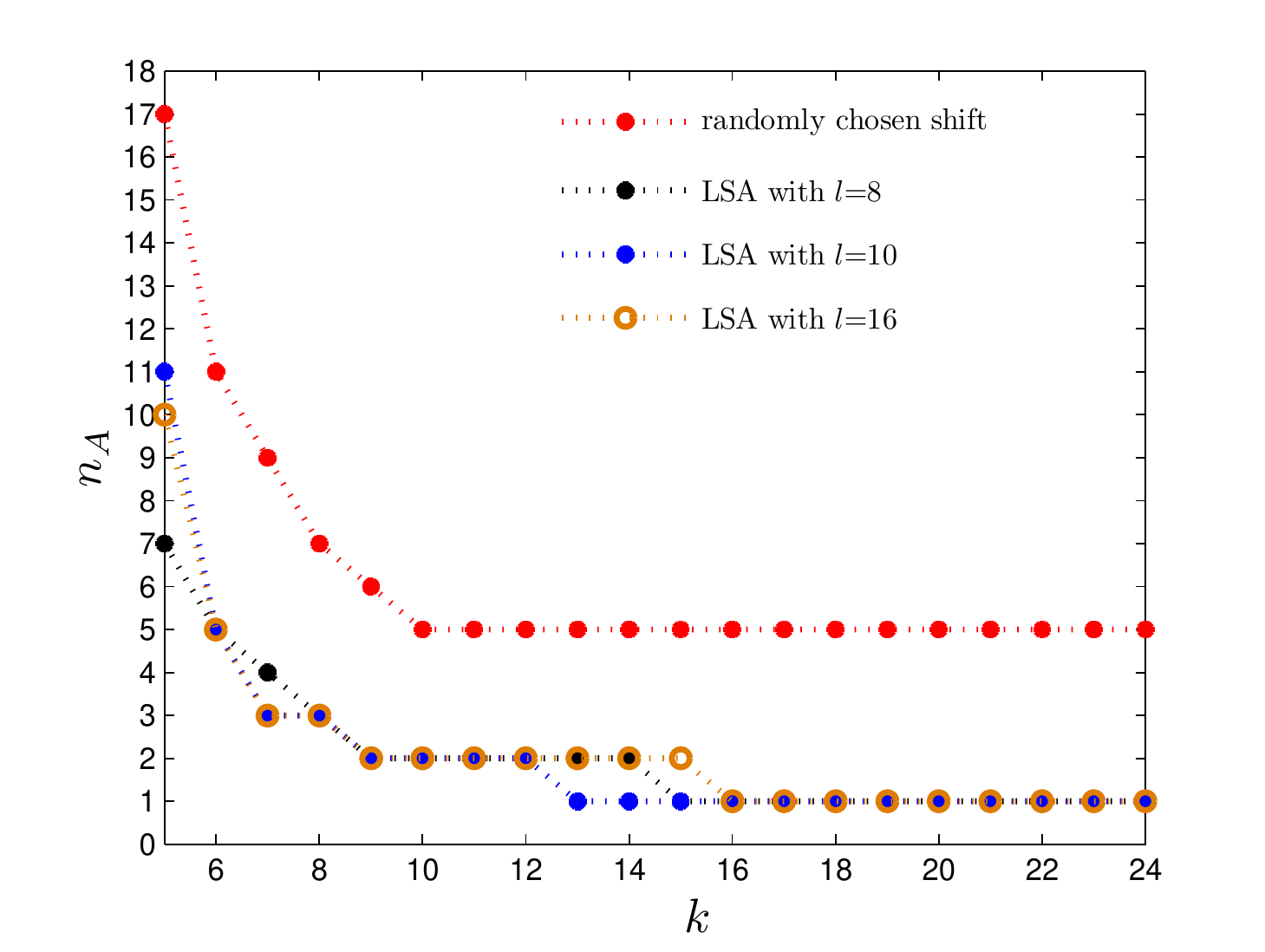}
  \caption{Effects of different shifts in compressing Alice's information. Here, $p=0.25$. When the shifts are restricted to $-l$ and $l$ (with $l=8,10$, or $16$), $n_A$ is reduced to 1 quickly with the increase of $k$, but if the shifts are randomly chosen, it is very difficult to reduce the number of Alice's known final bits to less than 5, which shows that the ``low-shift'' is essential in realizing ideal database security.}
\end{figure}
Finally, we emphasize the following two points about our generic QOKT-PQ model with the ``IDS-ZF'' feature.
\begin{itemize}

\item The raw oblivious key must be random, that is, the bits Alice knows must be at random positions and she cannot disturb the randomness, otherwise the LSA would be invalid in compressing Alice's known bits. For example, suppose Alice can make the first $pl$ bits of each block known to her and keep the remaining ones unknown, then after the LSA, she would know about $pN$ final key bits by announcing that each shift equals to 0. Therefore, the randomness of the raw key is necessary. Luckily, it would not affect the universality of our model because some skillful methods, such as the ``shift-and-permutation'' in \cite{Wei1}, can be used to ensure the randomness of the oblivious key effectively.

\item Block length $l$ should be selected from two aspects. First, it should restrict $p_{discard}$ to a small value. A small $p_{discard}$ means a small amount of discarded bits and a saving in communications. More importantly, Alice may obtain benefit in arranging the alignment of the blocks in LSA by telling a lie that certain blocks are completely unknown and should be discarded. As the proportion of discarded blocks should be about $p_{discard}$, this benefit can be restricted to a small extent if $p_{discard}$ is small. Second, proper $l$ should ensure a relative high efficiency of compressing Alice's information. As Table 2 shows, $\overline{k}_{l}$ may increase with the growth of $l$, which clearly means a reduction in the compression efficiency. When $p=0.25$, $l=10$ is a better choice compared to $l=8$ (which induces a high $p_{discard}$) and $l=16$ (which results in an obvious reduction in compression efficiency.)
\end{itemize}

\section{Conclusion}
With the aid of the LSA technique, we present an improved RRDPS-PQ protocol which not only can be used to retrieve from large database but also retains the merits of ``IDS-ZF'' even under imperfect source. Inspired by this, we generalize the LSA technique and construct a generic QOKT-PQ model with the ``IDS-ZF'' feature, thus breaking the troublesome stalemate of ``database security'' and ``failure probability'' in general QOKT-PQ protocols completely.

\begin{acknowledgments}
This work was supported by the National Natural Science Foundation of China (Grant Nos. 61672110, 61572246, 61602232, 61202317), the Open Foundation of State key Laboratory of Networking and Switching Technology (Grant No. SKLNST-2016-01), the Plan for Scientific Innovation Talents of Henan Province(Grant No. 164100510003), the Program for Science \& Technology Innovation Talents in Universities of Henan Province (Grant No. 13HASTIT042), and the Key Scientific Project in Universities of Henan Province (Grant Nos. 16A120007, 16A520021), the National Postdoctoral Program for Innovative Talents (Grant No. BX201600199) and the Fundamental Research Funds for the Central Universities (Grant No. 106112016CDJXY180001).
\end{acknowledgments}

\end{document}